\documentclass{article}
\usepackage{dcolumn}% Align table columns on decimal point
\usepackage{bm}% bold math
\usepackage{verbatim}       % Defines \begin{comment}, \end{comment}

\usepackage[pdftex]{graphicx}
\usepackage{amsthm}
\usepackage{amssymb}
\usepackage{bm}
\usepackage{amstext}
\usepackage{psfrag}
\usepackage{amsmath}
\usepackage{amsfonts}
\usepackage{mathrsfs}
\usepackage{cite}
\usepackage{color}

\newfont{\bb}{msbm10 at 12pt}

\newcommand{\p}{\partial}

\newcommand{\dd}{{\rm d}}
\newcommand{\bd}{\begin{definition}}                %inizia definizione
\newcommand{\ed}{\end{definition}}                  %fine definizione
\newcommand{\bc}{\begin{corollary}}                 %inizia corollario
\newcommand{\ec}{\end{corollary}}                   %fine corollario
\newcommand{\bl}{\begin{lemma}}                     %inizia lemma
\newcommand{\el}{\end{lemma}}                       %fine lemma
\newcommand{\bp}{\begin{proposition}}            %inizia proposizione
\newcommand{\ep}{\end{proposition}}                %fine proposizione
\newcommand{\bere}{\begin{remark}}                  %inizia osservazione
\newcommand{\ere}{\end{remark}}                     %fine oservazione

\newcommand{\bt}{\begin{theorem}}
\newcommand{\et}{\end{theorem}}

\newcommand{\be}{\begin{equation}}
\newcommand{\ee}{\end{equation}}

\newcommand{\bit}{\begin{itemize}}
\newcommand{\eit}{\end{itemize}}
\newtheorem{theorem}{Theorem}[section]
\newtheorem{corollary}[theorem]{Corollary}

\newtheorem{lemma}[theorem]{Lemma}
\newtheorem{proposition}[theorem]{Proposition}
\theoremstyle{definition}
\newtheorem{definition}[theorem]{Definition}
\theoremstyle{remark}
\newtheorem{remark}[theorem]{Remark}

\begin{document}
%
%\DeclareGraphicsExtensions{.pdf}

\title{Some regularity results for Lorentz-Finsler spaces}

\author{E. Minguzzi\footnote{Dipartimento di Matematica e Informatica ``U. Dini'', Universit\`a degli Studi di Firenze,  Via
S. Marta 3,  I-50139 Firenze, Italy. E-mail:
ettore.minguzzi@unifi.it} \ \ and \  \ S. Suhr\footnote{University of Bochum. E-mail: Stefan.Suhr@ruhr-uni-bochum.de\newline Stefan Suhr is supported by the SFB/TRR 191 ``Symplectic Structures in Geometry, Algebra and Dynamics'', funded by the Deutsche Forschungsgemeinschaft. }}

\date{}

\maketitle

\begin{abstract}
\noindent We prove that  for continuous Lorentz-Finsler spaces timelike completeness implies inextendibility. Furthermore, we prove that under suitable locally Lipschitz conditions on the Finsler fundamental function the continuous causal curves that are locally length maximizing (geodesics) have definite  causal character, either lightlike almost everywhere or timelike almost everywhere. These results generalize previous theorems by Galloway, Ling and Sbierski, and by Graf and Ling.
\end{abstract}

%\pacs{}

%\noindent Key Words:

\section{Introduction}
Recently, Galloway, Ling and Sbierski \cite{galloway18b} have proved the following result

\begin{theorem} \label{msd}
A smooth (at least $C^2$) Lorentzian spacetime that is timelike
geodesically complete and globally hyperbolic is $C^0$-inextendible.
\end{theorem}

In this work we are going to improve this result in three directions: (a) in the regularity of the starting manifold,  (b) in the possible (Finslerian) anisotropy of the spacetimes considered, and (c) in the causality of the starting manifold (we place no condition). In fact we prove the next result

\begin{theorem} \label{jif}
A $C^0$ proper Lorentz-Finsler space which is asymptotically timelike geodesically complete is inextendible as a $C^0$ proper Lorentz-Finsler space.
\end{theorem}

This theorem seems optimal for what concerns the category of Lorentz-Finsler spaces. We mention that the issue of inextendibility has also been approached using other tools pertaining to low regularity geometry. In fact Grant, S\"amann and Kunzinger  have obtained an inextendibility result \cite[Theorem 5.3]{grant18} analogous to  Theorem  \ref{jif} by using the theory of Lorentzian length spaces \cite{kunzinger18}. Their result also does not assume global hyperbolicity but, contrary to our own, it does not imply  Theorem \ref{msd} since when restricted to Lorentzian manifolds it demands Lipschitz metrics on the original and extended spacetimes and strong causality of the original spacetime. However, it covers  other non regular circumstances that can only be approached through the formalism of Lorentzian length spaces, e.g.\ non-$C^1$ manifolds.

Graf and Ling \cite{graf18} have also obtained a generalization of Theorem  \ref{msd}. They  removed the global hyperbolicity assumption at the price of a strengthening of the regularity condition on the original and extended spacetimes, that is, they demand the metrics to be Lipschitz. Their result is also implied by our Theorem \ref{jif} or by Theorem  \ref{jjf} below.

Finally, the reader is referred to \cite{sbierski15,galloway17b,chrusciel18} for inextendibility  results  that do not use the assumption of timelike completeness.

We need to clarify much of the terms entering  Theorem \ref{jif}. The notion of {\em closed Lorentz-Finsler space} has been recently introduced in \cite{minguzzi17} and can be regarded as  a generalization of the concept of  (regular, non-degenerate) {\em closed cone structure} \cite{bernard16,minguzzi17} aimed at providing the weakest and most natural conditions on the Finsler fundamental function $F$ in order to obtain most metric results of causality theory.

More precisely, let $x\mapsto C_x\subset TM\backslash 0$ be a distribution of non-empty closed sharp convex cones. A {\em closed cone structure} is such a cone distribution for which $C=\cup_x C_x$ is a closed subset of the slit tangent bundle $TM\backslash 0$. This closure property is equivalent to the upper semi-continuity of the cone distribution \cite[Prop.\  2.3]{minguzzi17} \cite[Thm.\ 1.1.1, Prop.\ 1.1.2]{aubin84}.
Moreover, it is called a {\em proper cone structure} if  additionally its interior is non-empty over every fiber $(\textrm{Int} C)_x\ne \emptyset$, or equivalently if it contains a continuous distribution of proper cones (closed, sharp, convex, with non-empty interior). The set $(\textrm{Int} C)_x$ is the {\em timelike cone} of the theory. A $C^0$ cone structure  is one for which the cone distribution is continuous and in this case $(\textrm{Int} C)_x=\textrm{Int} C_x$, cf.\ \cite[Proposition  2.6]{minguzzi17}.

Let us consider a closed cone structure for which a concave, positive homogeneous function $F\colon C\to [0,+\infty)$ has been given. Furthermore, suppose  that $F^{-1}(0)=\p C$ (this condition is not imposed in \cite{minguzzi17} thus the notion of  Lorentz-Finsler space used in this work is slightly more restrictive than that considered in that work). Let us define a cone distribution on $M^\times=M\times \mathbb{R}$ through
\begin{equation} \label{nhz}
C^\times_{P}=\{(y,z) \colon y\in C_p, \ z\in T_r \mathbb{R}, \ \vert z\vert \le F(y) \}.
\end{equation}
where $P=(p,r)$. Following \cite{minguzzi17} we say that  $(M,F)$ is a  {\em  closed/proper/$C^0$ Lorentz-Finsler space} if $(M^\times,C^\times)$ is a closed/proper/$C^0$ cone structure.
%Similarly, we say that  $(M,F)$ is a  {\em  proper Lorentz-Finsler space} if $(M,C^\times)$ is a proper cone structure.
The $C^0$ Lorentzian spacetimes are perhaps the simplest examples of $C^0$ proper Lorentz-Finsler spaces \cite[Theorem 2.51]{minguzzi17}. A closed/proper/$C^0$-proper  Lorentz-Finsler space induces a closed/proper/$C^0$-proper cone structure $(M,C)$.

A {\em continuous causal curve} is an absolutely continuous curve with causal tangent vector almost everywhere. It becomes Lipschitz whenever parametrized with respect to $h$-arc length, where $h$ is any Riemannian metric. A {\em timelike curve} is a   piecewise $C^1$ continuous causal curve with timelike tangent.

In a closed Lorentz-Finsler space the Lorentz-Finsler distance is defined in the usual way. The length of a continuous causal curve $x\colon [0,1]\to M$, is
\[
\ell(x)=\int_0^1 F(\dot x) \dd t.
\]
The causal relation $J$ is the set of pairs $(p,q)$ such that $p=q$ or there is a continuous causal curve connecting $p$ to $q$. The chronological relation $I$ is the set of pairs $(p,q)$ such that  there is a timelike curve connecting $p$ to $q$. Clearly, $I\subset J$. The chronological relation is open.

The (Lorentz-Finsler) distance is defined by: for $(p,q) \notin J$, we set $d(p,q)=0$, while for $(p,q)\in J$
\begin{equation}
d(p,q)=\textrm{sup}_x \ell(x) ,
\end{equation}
where $x$ runs over the continuous causal curves which connect $p$ to $q$.

In closed Lorentz-Finsler spaces every point admits globally hyperbolic neighborhoods \cite[Proposition  2.10]{minguzzi17}, the length functional is upper semi-continuous \cite[Theorem 2.54]{minguzzi17}, limit curve theorems hold true, and under global hyperbolicity any two points are connected by a maximizing continuous causal curves \cite[Theorem 2.55]{minguzzi17}.

%\begin{remark}
%We stress that these definitions and results do not use the condition $F^{-1}(0)=\p C$ which, however, will be tacitly assumed from now on. It is understood that this assumption is contained in the definition of proper/closed Lorentz-Finsler space used in this work, and hence also in {\color{blue} Theorem \xout{Thm.}}\ \ref{jif}.
%\end{remark}

In a closed Lorentz-Finsler space several definitions of causal geodesic are  possible. In this work we shall be concerned with locally length maximizing continuous causal curves. We shall call them {\em causal maximizers} so avoiding the term {\em geodesic}. Only at the end of the paper we shall discuss  whether these curves  coincide with the geodesics as they have been defined in other papers.

Under the low regularity conditions of this work we do not have a notion of affine parameter at our disposal. Therefore, we need a suitable notion of geodesic completeness.
\begin{definition}
A closed Lorentz-Finsler space is {\em future asymptotically timelike geodesically complete} if every causal maximizer $\gamma\colon [0,a)\to M$  that (a) escapes every compact set, and (b) has positive length, has actually infinite length.
\end{definition}

Finally, a closed Lorentz-Finsler space $(M',F')$  is an {\em extension} of the closed Lorentz-Finsler space $(M,F)$ if they have the same dimension and there is a $C^1$ embedding $\varphi \colon M\to M'$ such that $\varphi_*C=C'$, $\varphi^*F'=F$ and $\p \varphi(M)\ne \emptyset$. The spacetime $(M,F)$ is  {\em inextendible} if it has no extension.
We have so completed the introduction of the elements entering  Theorem  \ref{jif}.

It is certainly pleasing when the  causal maximizers have a definite causal character, namely a tangent which is almost everywhere timelike or almost everywhere lightlike. In fact, under this property one has  {\em future asymptotic timelike geodesic completeness} provided the causal maximizers $\gamma\colon [0,a)\to M$ which are  timelike and escape every compact set have infinite length.

%Notice that already in the regular case the proof would not demand the timelike completeness of timelike geodesics imprisoned in a compact set.

In the Lorentzian case Graf and Ling \cite{graf18} have shown that causal maximizers have a definite causal character under a Lipschitz condition on the metric. In Section \ref{nhd} we generalize their result to the Lorentz-Finsler case, by placing a Lipschitz condition on the fundamental Finsler function $F$, see Theorem \ref{jol}. Moreover, we prove that all definitions of {\em geodesic} proposed in the literature really coincide for locally Lipschitz Lorentz-Finsler spaces. Another consequence is

\begin{theorem} \label{jjf}
Timelike complete $C^0$ proper Lorentz-Finsler spaces which are  Lipschitz in the sense of Theorem  \ref{jol} cannot be extended as $C^0$ proper Lorentz-Finsler spaces.
\end{theorem}

As for our conventions,  in this work the manifolds $M$ and $M'$
are assumed to be connected, Hausdorff, second countable (hence paracompact) and of dimension $n+1$. Furthermore, they are $C^r$, $1\le r\le \infty$. The Lorentzian signature is $(-,+,\cdots,+)$. Greek indices run from $0$ to $n$. Latin indices from $1$ to $n$.  The Minkoski metric is denoted $\eta_{\alpha \beta}$, so $\eta_{00}=-1$, $\eta_{ii}=1$. The subset symbol $\subset$ is reflexive. The boundary of a set $S$ is denoted $\p S$.
%All our parametrized causal curves are future directed.
In order to simplify the notation we often use the same symbol for a curve or its image. A causal diamond is a set of the form $J^+(p)\cap J^-(q)$, and similarly for the chronological diamond where $I$ replaces $J$. Many statements of this work admit, often without notice, time dual versions obtained by reversing the time orientation of the spacetime.

\section{Inextendibility}
%Following the strategy of \cite{galloway18b} we are going to prove the next result

%Our goal is to prove the next result
%\begin{theorem}
%If a proper Lorentz-Finsler space $(M,F)$ admits an extension as a proper Lorentz-Finsler space $(M',F')$, then there is a causal maximizer of positive length with endpoint in the boundary of $M$.
%\end{theorem}
For notational simplicity, in what follows we shall regard $M$ as a subset of $M'$ thus, for instance, we shall write $\p M$ in place of $\p \varphi(M)$.

We need to define the future and past boundaries.
\begin{definition}
Let $(M,F)$ and $(M',F')$ be proper Lorentz-Finsler spaces and let the latter be an extension of the former.
The {\em future boundary} $\p^+ M$ is the subset of $\p M$ which consists of endpoints $\gamma(1)$ of timelike curves $\gamma\colon [0,1]\to M'$, $\gamma([0,1))\subset M$. The {\em  past boundary} is defined dually.
\end{definition}

\begin{lemma}
Let $(M,F)$ and $(M',F')$ be proper Lorentz-Finsler spaces and let the latter be an extension of the former, then $\p^+M\cup \p^- M\ne \emptyset$. If $\p^+ M\ne \emptyset$
then for every $q\in \p^+ M$ and every neighborhood $U$ of $q$ we can find a future directed timelike curve
$\beta \colon [0,1]\to U$ with $\beta([0,1)) \subset M$, $\beta(1)\in \p^+ M$,
and a local (flat) Minkowski metric $\eta$ defined on a neighborhood of $\beta(1)$ inside $U$
such that (a) the  $\eta$-cones are contained in $\mathrm{Int} C'$,  (b) $\beta$ is a timelike $\eta$-geodesic, and (c) let $w$ be a parallel vector (for the flat affine connection induced by $\eta$) coincident with $\dot \beta$ at $\beta(1)$, then  $\eta(w,v)<0$ for every $v\in C'$.

A dual statement holds if $\p^- M\ne \emptyset$.
\end{lemma}

\begin{proof}
Let $x \in \p M$, let $R_x$ be a Lorentzian (round) cone such that $R_x \subset \textrm{Int} D_x$ where $D$ is a  $C^0$ proper cone structure $D\subset C'$ (which exists since $(M',F')$ is a {\em proper} Lorentz-Finsler space). In a   neighborhood of $x$ we can extend $R_x$ to the cone structure $R$ of a Minkowski (flat) metric $\eta$.
%Let $\{x^\alpha\}$ be  local coordinates such that $\eta$ takes its canonical form, in such a way that $\p_0$ is future directed timelike.
By continuity of $D$ the inclusion $R_z\subset \textrm{Int} D_z\subset C'_z$ holds in a neighborhood $V$ of $x$. As a consequence, $R\subset \textrm{Int} D\subset \textrm{Int} C$, cf.\ \cite[Proposition 2.6]{minguzzi17}.  At this point the proof proceeds as that of \cite[Lemma 2.17]{sbierski15}. Let $y \in I_\eta^-(x,V)$, with $y$ sufficiently close to $x$ such that $x$ and $y$ are connected by a timelike $\eta$-geodesic $\beta$ contained in $V$. If $y \in M$ then $\beta$ intersects $\p M$ because $x\notin M$ and we are done. If $y\notin M$, by the openness of $I^+_\eta(y,V)$ we can find some $x'\in I^+_\eta(y,V)\cap M$ and a timelike $\eta$-geodesic $\beta'$ connecting $y$ to $x'$, so that $\beta'$ intersects $\p M$. In both cases we have found a timelike $\eta$-geodesic intersecting both $M$ and $M'\backslash M$, thus by the openness of $M$ we can find a maximal connected timelike segment contained in $M$ with endpoint belonging to $\p M$ so that, due to the presence of the timelike segment, the endpoint really belongs to $\p^+ M$ or $\p^- M$. Hence $\p^+M\cup \p^- M\ne \emptyset$.
%is the point  in the statement of the lemma.

If $\p^+ M\ne \emptyset$, let $q\in \p^+ M$, and $U$ be a neighborhood of $q$. Further let $\gamma\colon [0,1]\to M'$ be  timelike curve with $\gamma([0,1)) \subset M$ and $\gamma(1)=q$. Since $\gamma$ is piecewise $C^1$, by shortening the domain  we can assume that $\gamma$ is $C^1$ timelike.
Since $\dot \gamma(1)$ is timelike it belongs to a $C^0$ proper cone structure $D\subset C'$.
Let $\hat R_q$ be a Lorentzian (round) cone such that $ \dot \gamma(1) \in \textrm{Int} \hat R_q \subset \hat R_q \subset \textrm{Int} D_q$. In a   neighborhood of $q$ we can extend $\hat R_q$ to the cone structure $\hat R$ of a Minkowski (flat) metric $\hat \eta$. By continuity of $D$ the inclusion $\hat R_z\subset  \textrm{Int} D_z  \subset D_z\subset C'_z$ holds in a neighborhood  $U'\subset U$  of $q$, thus $\hat R\subset \textrm{Int} C'$ on $U'$, cf.\ \cite[Proposition 2.6]{minguzzi17}. For $z\in \gamma\backslash \{q\}$ sufficiently close to $q$ we have $z\in I^{-}_{\hat \eta}(q)$. Consider the $\hat \eta$-timelike geodesic connecting $z$ to $q$. It starts from $M$ and escapes  from it to the future so a suitable restriction provides the sought curve $\beta\colon [0,1]\to M$.

Let $p=\beta(1)$ and let $\{\tilde x^\alpha\}$ (the tilde is for notational consistency with the next proof) be  affine coordinates (for the affine structure induced by $\hat \eta$) such that $\tilde x^\alpha(p)=0$ for every $\alpha$, $\dot \beta =\p/\p \tilde x^0$, $\dd \tilde x^0(v)>0$ for every $v\in C'_p$. The coordinates $\tilde x^i$, $i=1,\cdots,n$, can be rescaled with a shared factor so that the canonical metric $\eta=-(\dd \tilde x^0)^2+\sum_i (\dd \tilde x^i)^2$ has a cone $R_p\subset \textrm{Int} \hat R_p$.
By the upper semi-continuity of $C'$ the condition, $\dd \tilde x^0(v)>0$ for every $v\in C'$, is preserved in a neighborhood of $p$, and by continuity the condition $R_p \subset \textrm{Int} D_p$ extends to a neighborhood of $p$, hence $R\subset \textrm{Int} C'$. We conclude that (a)-(c) hold in a neighborhood of $p$.
\end{proof}
%
%
%The previous proof really shows that there is $p \in \p^+ M$ (not necessarily the point $p$ of the proof) and a timelike curve $\gamma\colon [0,1)\to M$  with future endpoint  $p$ (or time dually) where $\gamma$ is really a timelike geodesic for a local (flat) Minkowski metric $\eta$ whose cones are strictly contained in $C$ in a neighborhood $p$.

%For  definiteness we shall also assume that the boundary point $p$ belongs to $\p^+ M$ and that the coordinates are chosen so that $x^\alpha(p)=0$ for every $\alpha$.
%
%
%
%Consider the flow of $\p_0$ and its integral $C^1$ timelike curves in $U$. Since $p\in \p M$ some of them intersect $M'\backslash M$ while some of them intersect $M$. If some of them intersect both $M'\backslash M$ and $M$ we have finished. If  not then each of them either belongs to  $M'\backslash M$ or to $M$. In the latter case there is a whole neighborhood of the integral curve belonging to $M$ (since $M$ is a manifold) and through each point of the neighborhood passes an integral curve entirely contained in $M$. In other words $U\cap M$ is an open set generated by integral curves of $\p_0$. We can now consider a $\eta$-timelike geodesic passing through $p\in \p M$ that is slightly tilted with respect to $\p_0$ so as to enter $M$ in the past (this possibility should be clear thinking of the quotient space $U\backslash \p_0$). This curve intersects both $M'\backslash M$ (because $p$ belongs to this set) and $M$. (Notice that the flat coordinates mentioned in the above comment might be only affinely related to those used in this proof.)

Theorem \ref{jif} is immediate from the next result.

\begin{theorem} \label{ppo}
Let $(M,F)$ and $(M',F')$ be $C^0$ proper Lorentz-Finsler spaces and let the latter be an extension of the former.
If $\p^+ M\neq \emptyset$ then for every $q\in \p^+ M$ and every $M'$-neighborhood $U\ni q$ we can find a future directed causal maximizer $\sigma \colon [0,1]\to U$, $\sigma([0,1))\subset M$,  of  finite positive length, with  endpoint  $\bar p\in \p M\cap U$.

\end{theorem}

By the Lemma we know that $\p^+ M\cup \p ^- M\ne \emptyset$, thus under extendibility either this version or its time dual apply.

\begin{proof}
%We assume $\p^+ M\ne \emptyset$ the other case being analogous.
%there exists $p\in \p^+ M$ and
By the Lemma there exists a local flat metric $\eta$ and a future directed $\eta$-geodesic $\beta\colon [0,1)\to M$ with endpoint $p \in \p^+ M$ as there mentioned. In particular, the  $\eta$-cones are contained in $\mathrm{Int} C'$, so that future directed $\eta$-timelike vectors are $C'$-timelike.
%(this is not going to be the point $p$ of the statement of this theorem).
% be the future directed $\eta$-geodesic mentioned in the Lemma.
%Without lose of generality we can assume that $p\in \partial^+M$ and that $\beta$ is future pointing.}
Let $\{\tilde x^\alpha\}$ be local coordinates in an arbitrarily small   coordinate neighborhood $V \ni p$ for which $\eta$ takes its canonical form with
$\dot\beta= \tilde \p_0:=\p/\p \tilde x^0$, $\tilde x^\alpha(p)=0$ for every $\alpha$, and $\dd \tilde x^0(v)>0$ for every $v\in C'$ over the neighborhood (these are the coordinates constructed at the end of the previous proof). Observe that $\tilde \p_0$ is $\eta$-timelike and hence $C'$-timelike in $V$. Observe also that the condition $\dd \tilde x^0(v)>0$ for every $v\in C'$ implies that $\tilde x^0$ is locally a time function for $C'$.
The neighborhood $V$ can be redefined to be the chronological diamond for a (flat) Minkowski metric $\hat \eta$ (not coincident with that introduced in the previous proof) with cones wider than $C'$ but such that $\dd \tilde x^0(v)>0$ for every future directed $\hat \eta$-causal vector $v$ (here the upper semi-continuity of $C'$ is used once again). Hence $V$ is $C'$-globally hyperbolic \cite[Proposition 2.10]{minguzzi17}.

In these coordinates $p=(0,\bm{0})$. Let $o=(-\delta, \bm{0})$, $\delta>0$, so that $o\in \beta\cap M$. Let us consider a cylinder of radius $\rho$ and height  $2\rho$ centered at $o$, so small that it is contained in $M$. Let us introduce the compact pyramid  $K$ obtained as the convex hull  of $p$ and the horizontal section of the cylinder passing through $o$, i.e.\
\[
K=\textrm{conv}\left\{ (0,\bm{0})\cup \{-\delta\}\times \bar B^n(\bm{0},\rho) \right\}.
\]
Since $\p M\cap K$ is compact  the function $\tilde x^0|_{\p M\cap K}$ reaches a minimum $T>-\delta+\rho$ at a point $r=(T,\bm{x})\in \p M\cap K$.

\begin{figure}[ht]
\begin{center}
 %\psfrag{p}{ \!\!\!\! $p=(0,\bm{0})$} \psfrag{x}{ \!\!\!\!\!\!\!\!$\tilde x^0=0$}  \psfrag{y}{\!\!\! $\hat \eta$}   \psfrag{c}{\!\!\! $C'$} \psfrag{e}{ \!\!\!\!\!\! $\eta$} \psfrag{r}{\!\! $r$} \psfrag{m}{\!\!\!$\p M$} \psfrag{d}{\!\! $\sigma$} \psfrag{s}{\!\! $s$} \psfrag{b}{\!\! $\beta$} \psfrag{k}{\, $K$} \psfrag{o}{\!\! $o\!=\!(-\delta,\bm{0})$} \psfrag{q}{ \!\!\! $D$} \psfrag{g}{ \!\!\! $2\rho$} \psfrag{f}{ \!\!\!\!\!\!\!\!\!\!\!\!\!\!\!\!\!\!\!\!\!\!\!\!\!\!\!\!\!\!\!\!\! $\{-\delta\}\times \bar B^n(\bm{0},\rho)$}
\includegraphics[width=9cm]{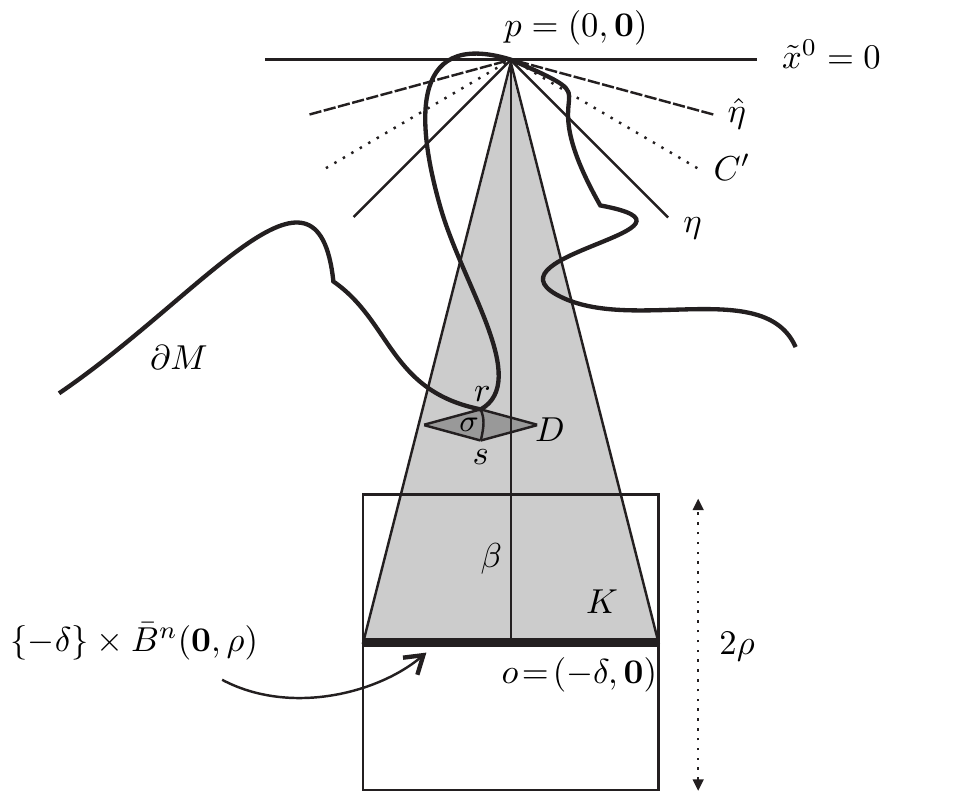}
\end{center}
\caption{The construction of $\sigma$ in case 1. } \label{ca1}
\end{figure}

We consider two cases:\\

Case 1, see Fig.\ \ref{ca1}. If $r \in \textrm{Int} K$ we can find a small $\hat \eta$-causal chronological diamond $D$ with lower vertex $s=(T',\bm{x})$, $-\delta<T'<T$, and upper vertex at $r$ which is entirely contained in $K$ (notice that $s$ and $r$ are connected by a $C'$-timelike curve $\nu:t\mapsto (t,\bm{x})$, $t\in [T',T]$, which is $C'$-timelike, because $\p_0$ is $\eta$-timelike, see the first paragraph of this proof). Notice that since $\tilde x^0$ is a local time function for $\hat \eta$ and $C'$ we have  $D\backslash \{ r\}\subset M$, so that
%{\color{blue}\xout{$J_{C'}^+(s,U)\cap J_{C'}^-(r,U)\backslash r\subset M$.}
\[
J_{C'}^+(s,V)\cap J_{C'}^-(r,V)\backslash \{r\}\subset M.
\]
 Since $V$ is $C'$-globally hyperbolic there is a maximizing continuous causal curve  \cite[Theorem 2.55]{minguzzi17}  $\sigma\colon [0,1]\to M'$, $\ell(\sigma)=d_{V}(s,r)$, of positive length between $s$ and $r$, so necessarily contained in $M$ except for the future endpoint $r$. Thus the length of the restriction $\sigma\vert_{ [0,1)}$ is positive. The point $r$ is the point $\bar p$ in the statement of the theorem.\\

Case 2, see Fig.\ \ref{ca2}. If $r\in \p K$ we introduce new local coordinates $\{ x^\alpha\}$ which are  affinely related to $\{\tilde x^\alpha\}$ and chosen in the following way: $x^\alpha(r)=0$ for every $\alpha$, $\p_0\propto \tilde{\p}_0$, $F_r(\p_0)=1$, and the locus $Q_r=\{ v\in T_rV \cap C': \dd x^0(v)=1\}$ is a subset of a support hyperplane at $\p_0$ for the convex set $F_r\ge 1$. From now on the Cartesian notation will refer to these coordinates. The important point is that the set $K$ is still a cone when seen in the new coordinates. Though no more isotropic it still contains $\beta$ so if the cone is regarded as the subset of a vector space we get that $-\p_0$ is still a direction, so to say, belonging to the interior of the cone. We find a round  (in the new coordinates) cone $\tilde K\subset K$ with vertex $r$ containing the direction $-\p_0$ in its interior, such that $\tilde K\backslash \{r\}\subset M$.  Let $\mu>0$ be the tangent of the angle formed by $\p\tilde K$ with respect to the vertical line generated by $\p_0$.

\begin{figure}[ht]
\begin{center}
 \includegraphics[width=8cm]{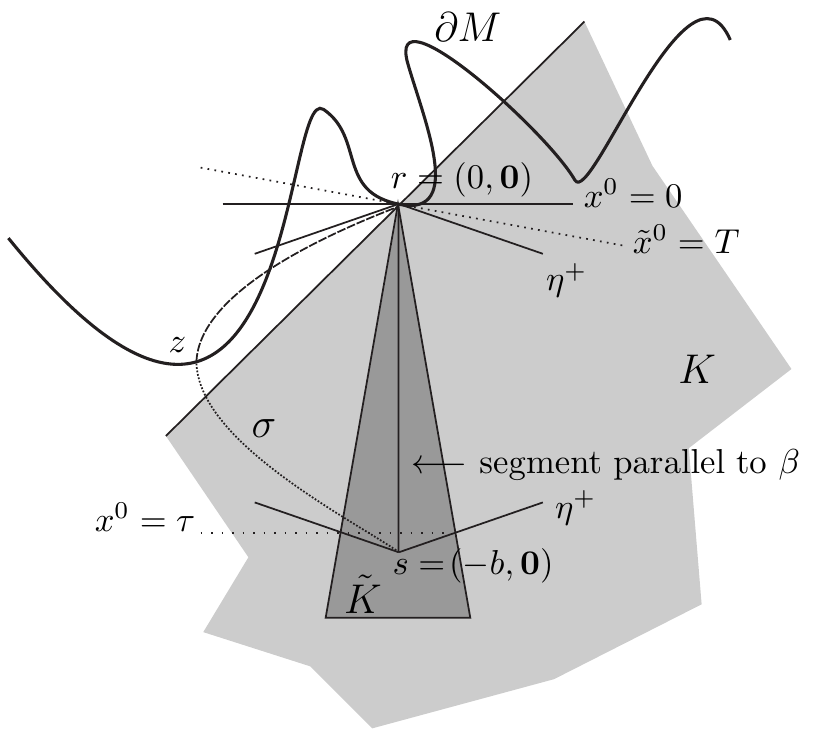}
\end{center}
\caption{The construction of $\sigma$ in case 2. } \label{ca2}
\end{figure}

Let us define the flat metric
$$\eta^+=-B^2 (\dd x^0)^2+\sum_i (\dd x^i)^2$$
%and
%$$\eta^-=-A^2 (\dd x^0)^2+\sum_i (\dd x^i)^2,$$}
where $B>0$.
The constant $B$ is chosen so large that $C'_r$ is contained in the timelike cone of $\eta^+$ at $r$. By the upper semi-continuity of $C'$ the inclusion is preserved in a neighborhood of $r$.
%, while $A$ is chosen so small that the causal cone of $\eta^-$ at $r$ is contained in the timelike cone of   $\textrm{C}'_r$.

    We can find a neighborhood $W$ of $r$ which is a chronological diamond for $\eta^+$, hence $C'$-globally hyperbolic.
    From now one we shall parametrize all the continuous causal curves in a neighborhood of $r$ with respect to $x^0$, so we are interested in the locus $Q=\{ v\in C'\cap TW\colon \dd x^0(v)=1\}$. Over $Q_r$ we have $F\le 1$ (cf.\ the mentioned support hyperplane), thus by upper semi-continuity of $C^\times$ (i.e.\ of $F$) for any chosen $0<\epsilon<1$ we can redefine $W(\epsilon)$ to be so small that $F<1+\epsilon$ over $Q$. Similarly, we have $F_r(\p_0)=1$ thus by the lower semi-continuity of $C^\times$ (i.e.\ of $F$) we can redefine $W$ so that  $F(\p_0)>1-\epsilon>0$ in $W$.

    Let $s=(-b, \bm{0})$ with $b>0$ so small that $s\in \textrm{Int} K\cap \textrm{Int} \tilde K \cap W$.
    Let us consider a maximizing continuous causal curve $\gamma$, $\ell(\gamma)=d_{W}(s,r)$, connecting $s$ to $r$ in the globally hyperbolic spacetime $(W,C')$. By the same argument used in the previous case $r\in I^+_{C'}(s,W)$ so the causal maximizer has positive length. In fact, since $F(\p_0)>1-\epsilon$ we have $\ell(\gamma)\ge (1-\epsilon) b$ where the right-hand side is a lower bound for the length of the coordinate-straight segment connecting the points. Let us consider the first point of escape $z$ of $\gamma$ from $M$, and let us cut $\gamma$ into two pieces, the curve $\sigma\colon [-b, -a) \to M$ starting from $s$ and with endpoint $z$ and the curve $\lambda\colon [-a, 0]\to M'$ connecting $z$ to $r$. If the latter were degenerate, i.e.\ $z=r$, we  would be finished since the former would have necessarily  positive length, so we shall assume $z\ne r$. Of course both curves are maximizing and our goal is to show that $\sigma$ has positive length (hence $\bar p=z$). We suppose not and we show that then $\ell(\gamma)$ has an upper bound which contradicts the previous lower bound.

     By construction $\tilde K\backslash \{r\} \subset M$ which implies that $z\notin \tilde K$. Moreover, since the cones $C'$ are contained in those of $\eta^+$, we shall have $x^0(z)\ge \tau$ where $\tau$ is determined by the equation $B(\tau-(-b))=\mu(0-\tau)$,  that is
    \[
    x^0(z)\ge -\frac{Bb}{B+\mu}  ,
    \]
     for the curve $\sigma$ is ``pushed forward to the future'' at least till it escapes $\tilde K$. The length of $\gamma$ satisfies
     \[
     \ell(\gamma)=\ell(\lambda)=\int_{x^0(z)}^0 F(v) \dd t\le (1+\epsilon) (0-x^0(z))\le (1+\epsilon) \frac{B}{B+\mu} b,
     \]
    where we used the fact that since $\lambda$ is parametrized with $x^0$, we have that the tangent vector belongs to $Q$ almost everywhere. We arrive at
    \[
    (1-\epsilon)b\le \ell(\gamma)\le (1+\epsilon) \frac{B}{B+\mu} b .
    \]
    Notice that $B$ and $\mu$ have been defined before the introduction of $\epsilon$. In other words these constants are independent of $\epsilon$. By taking $\epsilon$ (and hence $W$) sufficiently small we get a contradiction.
\end{proof}

An immediate consequence is the next generalization of \cite[Theorem 3.6]{galloway18b}

\begin{theorem}
Let $(M,F)$ be a $C^0$ proper Lorentz-Finsler space which is future asymptotically timelike geodesically complete and is extendible as a $C^0$ proper Lorentz-Finsler space, then $\p^+M=\emptyset$ and $\p^-M\ne \emptyset$.
\end{theorem}

We have also the next generalization of \cite[Corollary 3.7]{galloway18b} and \cite[Theorem 2.6]{galloway17b}. It has the same proof thanks to \cite[Theorem 2.19]{minguzzi17}.

\begin{theorem}
Let $(M,F)$ be a $C^0$ proper Lorentz-Finsler space extendible as a $C^0$ proper Lorentz-Finsler space. If $\p^+M=\emptyset$, then  $\p^-M$ is an achronal topological hypersurface (hence a locally Lipschitz graph).
\end{theorem}

\begin{remark}
Theorem \ref{ppo} immediately implies that there are no $C^0$ extensions at the boundary $\mathscr{I}^+$ of Minkowski spacetime, Schwarzschild's spacetime, or similarly  asymptotically flat $C^2$ spacetimes for which the timelike geodesics converge to $i^+$ (with no need to impose future timelike geodesic completeness). In fact, in this case the timelike geodesics escape the neighborhood $U$ without intersecting $\p M$. Of course, the proof of inextendibility at the spacelike boundary requires a different study \cite{sbierski15}.
\end{remark}

\section{Causal character of causal maximizers} \label{nhd}

The objective of this section is to prove the next theorem

\begin{theorem} \label{jol}
Let $x \mapsto C_x\subset T_xM\backslash 0$ be a continuous distribution of sharp convex closed cones with non-empty interior (continuous proper cone structure).
Let $C=\cup_x C_x\subset TM\backslash 0$ and let $F\colon C\to [0,\infty)$ be a continuous function which on every cone $C_x$ is positive homogeneous of degree one, concave, non identically zero and such that $F_x(\p C_x)=0$. Suppose there is a $C^1$ strictly convex function $f\colon [0,+\infty)\to [0,+\infty)$, $f(0)=0$, $f'(0)=0$,  such that   the function
$f(F)$ can be locally extended\footnote{The extended function is denoted in the same way. Whitney's Theorem \cite{whitney34} and similar results for Lipschitz functions could likely be used to place conditions on the function $f(F)$ just over $C$.} to a neighborhood of $C$ so as to have  the next properties:
\begin{itemize}
\item[(a)] it is locally Lipschitz in $x$ for bounded velocity cf.\ Eq.\ (\ref{fox}),
\item[(b)] the vertical differential map $(x,v)\mapsto \dd f(F_x)(v)(\cdot)$ is continuous and \\ $\dd f(F_x)\ne 0$ on $\p C_x$  (thus $\p C_x$ is $C^1$).
\end{itemize}
%\begin{itemize}
%\item[(a)] For every $x$, $F^\rho_x$ is $C^1$ on $C_x$ (including the boundary) and $\dd F^\rho_x\ne 0$ on $\p C_x$ (thus $\p C_x$ is $C^1$),
%\item[(b)] the function $F^\rho_x$, extended by setting $F^\rho_x=0$ outside $C_x$,  is locally Lipschitz in $x$,
%\end{itemize}
Then the maximizers of  the length functional $\ell(\gamma)=\int_\gamma F(\dot \gamma)\dd t$ are either timelike or lightlike, namely the tangent is timelike almost everywhere or lightlike almost everywhere.
\end{theorem}

The Lipschitz condition  takes care of the positive homogeneity of $F$ in the fiber direction, the reader is referred to the inequality  (\ref{fox}) for a clarification.

The proof really shows that in the latter case the tangent cannot be  timelike anywhere. This fact follows, more generally, from a result proved in \cite[Theorem 18]{minguzzi17}.

A typical $f$ could be $f(x)=x^\rho/\rho$ for $\rho>1$. General relativity corresponds to $\rho=2$  with $F_x(v)=\sqrt{-g_x(v,v)}$. For a different example with $\rho>1$, consider for instance $F_x=\{(\dd x^0)^\rho-\sum_i (\dd x^i)^\rho \}^{1/\rho}$.

The proof shows that the differential $\dd L_x(v)(\cdot)$, $L:=-f(F)$, does not vanish anywhere on $C$. This quantity is interpreted as the momenta of the particle of velocity $v$ \cite[Section 3.1]{minguzzi17}, so in $(b)$ we are asking that this correspondence velocity-momenta be continuous and well defined even for massless particles (i.e.\ every massless particle has finite and non-vanishing momenta). In fact, as established in \cite{minguzzi17}, the function $f$ determines the correspondence velocity-momenta by regularizing the Finsler Lagrangian $L$ at the boundary of the cone.

The first paragraphs and the general strategy are the same of Graf and Ling's proof, however, we need to work out the estimates entering it in a different way since we cannot appeal anymore to the Lorentzian Lemma by Chrusciel and Grant \cite[Lemma 1.15]{chrusciel12}.
%We changed the proof by  Graf and Ling only where it was strictly necessary, and
We kept the same notation of Graf and Ling, this way the Finslerian proof should be easier to follow by people already acquainted with the Lorentzian one.

\begin{proof}
Let  $\gamma\colon I \to M$ be a maximizing future directed causal curve from $p$ to $q$. By compactness we may cover $I$ by finitely many open intervals $I_k$ (half open intervals on the endpoints of $I$) such that each $\gamma_{I_k}$ is contained in a relatively compact chart domain $(U_k, \varphi_k)$ on which $F(\p_0^{\varphi_k})>c_k>0$ and the causal cone $C$ is contained in the open  cone of a flat Minkowski metric $\eta_k$ on $U_k$. The time function $t:=x^0$ of the flat Minkowski metric in standard coordinates is a strictly increasing function in the curve parameter. On $U_k$ we can reparametrize the curve $\gamma$ using $t$ as parameter, so that $\dd t(\dot \gamma)=\dot \gamma^0=1$ almost everywhere. We want to show that it is possible to pass to a refinement covering, denoted in the same way, such that on $U_k$ we can find a translationally invariant (in the affine structure induced by the coordinates) cone distribution $x\mapsto \tilde C_x$, such that $C_x \subset \textrm{Int} \tilde C_x$ for every $x\in U_k$, the extension of $f(F)$ mentioned in the statement of the theorem is well defined on $\tilde C$, and  there are $\delta,u_0>0$ such that for $0\le u<u_0$ we have
\begin{align} \label{dte}
f(F(v+u\p_0))\ge f(F(v))+ u \delta
\end{align}
for every $v\in \tilde D:=\{v\in T\bar{U}_k\cap \tilde C\colon \dd t(v)=1\}$. This shrinking result is accomplished as follows: we need only to prove that for $x\in \gamma(I_k)$ we have  $\dd f(F_x)(v)(\p_0)>0$ at every point of the compact set $D_x=\{v\in  C_x: \dd t(v)=1\}$. By continuity of $(x,v)\mapsto \dd f(F_x)(v)(\p_0)$ the same inequality holds in a tangent bundle neighborhood of $D_x$, so that the cone $\tilde C$ can be found.

%Since $F_x$ is positive homogeneous we have for $v\in \textrm{Int} C_x$,
%\begin{align*}
%\dd f(F_x)(v)(v)&=\lim_{\epsilon\to 0} \frac{1}{\epsilon} [f(F_x(v+\epsilon v))-f(F_x(v))]\\
%&=\lim_{\epsilon\to 0} \frac{1}{\epsilon} [f((1+\epsilon) F_x( v))-f(F_x(v))]=f'( F_x( v)) F_x(v)>0,
%\end{align*}
% thus $\dd f(F_x)(v)(\cdot) \ne 0$. For $v\in \p C_x$ the inequality $\dd f(F_x)(v)(\cdot) \ne 0$ is an assumption.
For any $v\in D_x\cap\textrm{Int} C_x$
\begin{align*}
f(F_x(v))+&uF_x(\p_0)f'(F_x(v)) \le f(F_x(v)+ uF_x(\p_0))\\
&\le f(F_x(v+u\p_0))=f(F_x(v))+\dd f(F_x)(v)(\p_0) u+ o(u),
\end{align*}
where the first inequality follows from the convexity of $f$ and the second inequality is the reverse triangle inequality (which follows from concavity and positive homogeneity \cite[Proposition 3.4]{minguzzi17}).
Thus simplifying $f(F_x(v))$, dividing by $u$, and letting $u\to 0$, we get $ \dd f(F_x)(v)(\p_0)\ge  F_x(\p_0) f'(F_x(v)) >0$, which proves positivity of $\dd f(F_x)(v)(\p_0)$ at least inside the  cone. By continuity $\dd f(F_x)(v)(\p_0)\ge 0$ for $v\in \p C_x$. However, we assume $F_x=0$ on $\p C_x$, thus $\ker \dd f(F_x)(v)$, for $v\in \p C_x$, has dimension at least $n$ as it includes $T\p C_x$, while it cannot have dimension $n+1$, for otherwise the differential would vanish on the boundary of the cone. Since $\p_0\notin T \p C_x$ we have $\dd f(F_x)(v)(\p_0)\ne 0$ and hence $\dd f(F_x)(v)(\p_0)> 0$.
%Now, since $F_x(v)=0$
%\begin{align*}
%0<u^\rho F^\rho_x(\p_0) = [F_x(v)+ uF_x(\p_0)]^\rho\le F_x^\rho(v+u\p_0)=F_x^\rho(v)+\dd F^\rho_x(v)(\p_0) u+ o(u),
%\end{align*}
%and we must necessarily have $\dd F^\rho_x(v)(\p_0)>0$ otherwise for sufficiently small $u$ we would get a contradiction.
%By assumption this derivative is different from zero, and in fact it must be positive because $F(v+u\p_0)\ge F(v)+ uF(\p_0)$ by the reverse triangle inequality (which follows from concavity and positive homogeneity \cite[{\color{blue} Proposition \xout{Prop.}}\ 3.4]{minguzzi17}).
The inequality (\ref{dte}) is proved. Notice that the inequality $\dd f(F_x)(v)(\p_0)\ne 0$ for any $v\in \p C_x$ also implies that $\tilde C$ can be taken so small that $F(v)<0$ iff $v\notin C$.

%Let us introduce the compact set $D=\{v\in T\bar{U}\cap C: \dd t(v)=1\}$. We want to prove that  there are $\delta,u_0>0$ such that for $0\le u<u_0$ we have
%\begin{align} \label{ste}
%F^\rho(v+u\p_0)\ge F^\rho(v)+ u \delta
%\end{align}
%for every $v\in D$.

%Let $x$ be in the image of $\gamma$. We want to show that there is a neighborhood

Let $N_I := \{ s\in I\colon \dot\gamma(s) \textrm{ exists and is lightlike} \}$. Let $\mu$ be the Lebesgue measure of the real line. By contradiction, suppose that $0<\mu(N_I)<\mu(I)$, namely the maximizing curve is neither timelike almost everywhere nor lightlike almost everywhere.
%Suppose
%$\gamma\colon I \to M$ is a maximizing future directed causal curve from $p$ to $q$ which is not lightlike almost everywhere. We will prove that it  is  timelike almost everywhere. Suppose that it is not. Seeking a contradiction, suppose the set $N_I := \{ s\in I : \dot\gamma \textrm{ exists and is lightlike} \}$ has positive measure.
Below we will construct another causal curve from $p$ to $q$ which is longer than $\gamma$, and hence contradicting the fact that $\gamma$  is a maximizer.

%To do this, we first want to localize the situation. By compactness we may cover $I$ by finitely many open intervals $I_k$ (half open intervals on the endpoints of $I$) such that each $\gamma_{I_k}$ is contained in a relatively compact chart domain $(U_k, \varphi_k)$ on which $F(\p_0^{\varphi_k})>c_k>0$ and the causal cone $C$ is contained in the open  cone of a flat Minkowski metric $\eta_k$ on $U_k$.
%We are now going to show that
First let us show that for  at least one of the $I_k$, we have $0 < \mu(N_{I_k} ) < \mu(I_k)$ (i.e., $\gamma_{I_k}$ is causal but neither timelike nor null). Assume for the moment that $\mu (N_{I_j}) = 0$ for some $j$. Then, since the intersection of the neighbouring intervals $I_{j-1}$ and $I_{j+1}$ with $I_j$ must be
non-empty and open, either one of those has the desired property or $\mu(N_{I_{j-1}}) = \mu(N_{I_{j+1}}) = 0$.  The existence of a suitable $I_k$ now follows by induction and noting that $\mu(N_I) \ne 0$. If instead $\mu(N_{I_j}) = \mu(I_j)$, one proceeds the same way, using $\mu (N_I ) \ne \mu(I)$ in the end.

This shows that we may assume  that $\gamma(I)$ is contained in a chart domain $(U, \varphi)$. By reparametrizing  we may further assume that $\dot \gamma(0)$  exists and is timelike.
%Using a linear change of coordinates corresponding to a Gram-Schmidt orthogonalization process of $\{\p^\varphi_0, \ldots,\p^\varphi_{n-1}\}\vert_{\gamma(0)}$ and a translation we get new coordinates  $\psi$ on $U$ for which $\p^\psi_0\propto \p^\varphi_0$ and hence $F(\p^\psi_0)<0$, $\gamma^\psi(0)=0$.

%The causal cone $C_p$ is contained in a round cone $R_p$ of a flat Minkowski metric. By continuity the same is true in a neighborhood of $p$, thus
%The time function $t:=x^0$ of the flat Minkowski metric in standard coordinates is a strictly increasing function in the curve parameter. We can reparametrize the curve $\gamma$ using $t$ as parameter, thus $\dot \gamma^0(0)=1$.

To sum up, we need only to consider the case $M = \mathbb{R}^{n+1}$. The cone $C$ contains in its interior $\p_0$ and is contained in the Minkowski standard coordinate cone,  $\gamma \colon  [a, b] \to U \subset \mathbb{R}^{n+1}$, $\gamma(0)=0$, $\gamma$ is differentiable at $0$ and timelike, i.e.\ $F(\dot \gamma(0))>0$, $\mu(N_{[a,0]}) > 0$ (if instead $\mu(N_{[0,b]}) > 0$ one just needs to reverse the time orientation). Moreover, there is a translationally invariant cone distribution $\tilde C$ containing $C$ in its interior, such that the inequality (\ref{dte}) holds true.

Let $\gamma_1:=\gamma\vert_{[a,0]}$; given a $C^1$, hence Lipschitz function $h\colon [a,0]\to (0,\infty)$ we define a new curve $\Gamma_1^\mu=\gamma_1^\mu+\epsilon h T^\mu$, where $T^\mu$ is defined via $\p_0=T^\mu \p_\mu$ (i.e.\ $T^0=1$ and $T^i=0$, for $i\ne 0$). We want to show that there is  a suitable $h$ such that for sufficiently small $\epsilon$  we have $\Gamma_1\subset U$,  $\Gamma_1(0)=(\epsilon h(0), 0,\ldots,0)$, and  $F_{\Gamma_1}(\dot \Gamma_1) > F_\gamma(\dot \gamma)$ a.e.\ on $I$ (which will also imply that $\Gamma_1$ has causal tangent almost everywhere). Notice that $\Gamma_1^\mu(a)=\gamma^\mu_1(a)+\epsilon h(a) T^\mu$, with $h(a)>0$, thus it will be possible to join $\gamma^\mu_1(a)$ to $\Gamma^\mu_1(a)$ with a curve with tangent $\p_0$ hence timelike. This curve is denoted $\Gamma_0$ and has positive length, $L(\Gamma_0)> 0$.

%\begin{align*}
%F_{\Gamma_1}(\dot \Gamma_1) -F_\gamma(\dot \gamma)=[F_{\Gamma_1}(\dot \Gamma_1) -F_\gamma(\dot \Gamma_1)]+[F_\gamma(\dot \Gamma_1)-F_\gamma(\dot \gamma)]
%\end{align*}
%By the reverse triangle inequality we have
%\[
%F_\gamma(\dot \Gamma_1)-F_\gamma(\dot \gamma) \ge \epsilon f F(\p_0)
%\]
%RIFACCIAMO

It is convenient to write
\begin{align*}
f(F_{\Gamma_1}(\dot \Gamma_1))\! - \!f(F_{\gamma_1}(\dot \gamma_1))=[f(F_{\Gamma_1}(\dot \Gamma_1)) \!- \! f(F_{\Gamma_1}(\dot \gamma_1))]+[f(F_{\Gamma_1}(\dot \gamma_1))\!-\! f(F_{\gamma_1}(\dot \gamma_1))] .
\end{align*}
Notice that $\dot \gamma_1(t)$ might not belong to $C_{\Gamma_1(t)}$, but it certainly belongs to $\tilde C_{\Gamma_1(t)}$, due to the translational invariance of this set.
By Eq.\ (\ref{dte}) we have for sufficiently small $\epsilon$
\[
f(F_{\Gamma_1}(\dot \Gamma_1)) \ge f(F_{\Gamma_1}(\dot \gamma_1))+\epsilon \dot h \delta.
\]
%We have
%\begin{align*}
%F^2_{\Gamma_1}(\dot \Gamma_1) -F^2_{\Gamma_1}(\dot \gamma_1) & \ge \{F_{\Gamma_1}(\dot \gamma_1)+\epsilon \dot f\delta\}^2 -F^2_{\Gamma_1}(\dot \gamma_1) \\
%&\ge  2\epsilon \dot f F_{\Gamma_1}(\p_0)F_{\Gamma_1}(\dot \gamma_1) + \epsilon^2 \dot f^2 F^2_{\Gamma_1}(\p_0)\\
%& \ge 2\epsilon \dot f F_{\Gamma_1}(\p_0)F_{\Gamma_1}(\dot \gamma_1)
%\end{align*}
Moreover, by hypothesis $f(F)$ is locally Lipschitz, namely using the  coordinate Euclidean norm, for any $K>0$, we can find $k(K)>0$ such that  for every  $v$ in the extended domain of $F$ (namely $\tilde C$), $\vert v\vert\le K$,
\begin{equation} \label{fox}
\vert f(F_{y}(v))-f(F_x(v))\vert \le k \vert y-x \vert.
\end{equation}
Thus since any vector belonging to the compact set $\tilde D$ has norm bounded by some constant $K>0$ we can  find $k$ such that
\begin{equation} \label{fot}
\vert f(F_{\Gamma_1}(\dot \gamma_1))-f(F_{\gamma_1}(\dot \gamma_1))\vert \le k \vert\Gamma_1-\gamma_1 \vert = k h \epsilon  .
\end{equation}
%Here $\dot \gamma_1$ belongs to the compact set $\tilde D$, thus there is a constant $N>0$ such that $\vert v\vert <N$ over $\tilde D$.
Thus
\begin{align*}
f(F_{\Gamma_1}(\dot \Gamma_1)) -f(F_{\gamma_1}(\dot \gamma_1))\ge  \epsilon [\dot h \delta - h k ]\ge \epsilon r e^{\sigma t} \ge \epsilon r e^{\sigma a}>0 .
\end{align*}
The last inequalities follow taking $h=r e^{\sigma t}$ with $\sigma=\frac{1}{\delta}[1+ k ]$.
The coordinate definition of $\Gamma_1$ shows that this curve is absolutely continuous, while the previous inequality proves that it  has a causal tangent vector almost everywhere, namely it is a continuous causal curve.

Next the length of $\Gamma_1$ is bounded as follows
\begin{align}
\ell(\Gamma_1)&=\int_{ a}^0F_{\Gamma_1}(\dot \Gamma_1)\dd t=\int_{N_{[ a,0]}} F_{\Gamma_1}(\dot\Gamma_1)\dd t+\int_{[a,0]\backslash N_{[a,0]}} F_{\Gamma_1}(\dot \Gamma_1)\dd t \nonumber\\
&\ge f^{-1}(\epsilon r e^{\sigma a} )\mu(N_{[a,0]})+ \int_{[a,0]} F_{\gamma_1}(\dot \gamma_1) \dd t \nonumber \\
&\ge f^{-1}(\epsilon r e^{\sigma a} )\mu(N_{[a,0]})+\ell(\gamma_1) , \label{nis}
\end{align}
where we used the fact that $f$ is  strictly increasing hence invertible.

\begin{figure}[ht]
%\psfrag{y}{ \!\!\! $t=x^0$} \psfrag{t}{ \!\!\! $\tau_\epsilon$} \psfrag{j}{\!\!\!\! $\gamma(\tau_\epsilon)$}   \psfrag{u}{\!\!\!$\tau_\epsilon(1-\tfrac{\beta}{\alpha})$} \psfrag{c}{ \!\!\!\!\!\! $\Gamma_2$} \psfrag{h}{\!\! slope is  $\beta$} \psfrag{m}{\!\!\! slope is  $\alpha$} \psfrag{b}{\!\!\!\! $\gamma_2$} \psfrag{r}{\!\!\! $\gamma(0)$} \psfrag{x}{\! $\{x^i\}$} \psfrag{e}{\!\!\!\!\!\! $\Gamma_1$} \psfrag{d}{\!\!\! $\gamma_1$} \psfrag{f}{\!\!\!\!\!\! $\Gamma_0$} \psfrag{g}{\!\!\!\!\! $\gamma_1(a)$}
\begin{center}
 \includegraphics[width=10cm]{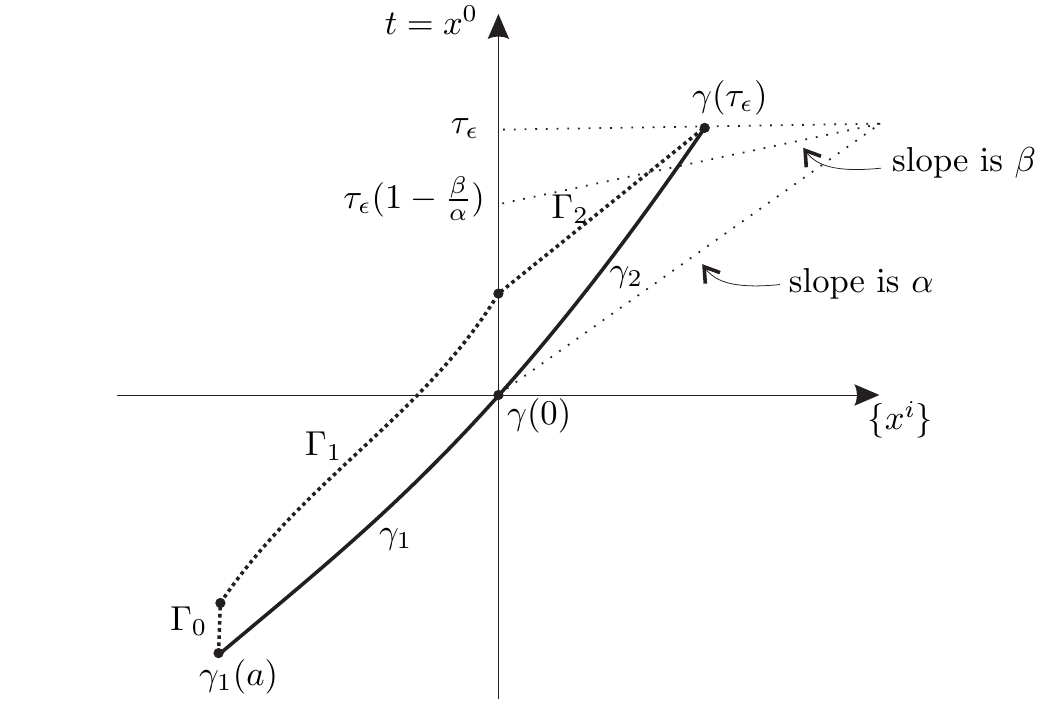}
\end{center}
\caption{The deformed curve $\Gamma$ is longer than $\gamma$, particularly due to the length of $\Gamma_1$. } \label{hra}
\end{figure}

We now turn to $\gamma\vert_{[0,b]}$ where we keep the $x^0$ parametrization. Since $\Gamma_1(0) \ne \gamma(0)$ we have to try to find $\tau_\epsilon > 0$ such that there exists a future directed causal curve $\Gamma_2$ from $\Gamma_1(0) = (\epsilon h(0), 0)$ to $\gamma_{\tau_\epsilon}= (\tau_\epsilon,\bar \gamma_{\tau_\epsilon})$ (see Fig.\ \ref{hra}) and such that for $\epsilon$ small enough
\[
\ell(\Gamma_0)+ \ell(\Gamma_1) +  \ell(\Gamma_2) > \ell(\gamma_1) + \ell(\gamma_2)
\]
where $\gamma_2 := \gamma\vert_{[0,\tau_\epsilon]}$. It should be noted
that the segment $\gamma_2$ of $\gamma$ depends on $\tau_\epsilon$ and hence depends on $\epsilon$ itself. It suffices to show $\ell(\Gamma_1) > \ell(\gamma_1) +\ell(\gamma_2)$. And using (\ref{nis}) we see that this holds if $\ell(\gamma_2)$ has an upper bound that scales with $\epsilon$ faster than  $f^{-1}(\epsilon r e^{\sigma a} )$. For instance, a linear bound would suffice, because $\lim_{\epsilon \to 0} f^{-1}(\epsilon r e^{\sigma a} )/\epsilon =+\infty$, as $f^{-1}(\epsilon r e^{\sigma a} )$ is locally at 0 lower bounded by any linear function.

We first consider the problem of bounding $\ell(\gamma_2)$, and then that of constructing $\Gamma_2$.

We know that $C$ is contained in the open cone of a flat Minkowski metric $\eta=(-\dd x^0)^2+\sum_i (\dd x^i)^2$, thus we can find a constant $K$ such that on $C_{\gamma(0)}$,
 % {\color{blue}\xout{$F < K\sqrt{(\dd x^0)^2-\sum_i (\dd x^i)^2}\le K \dd t$.}
 \[
 F < K\sqrt{(\dd x^0)^2-\textstyle{\sum}_i (\dd x^i)^2}\le K \dd t.
 \]
 By continuity the same bound holds on a neighborhood of $\gamma(0)$, thus for sufficiently small $\epsilon$, %{\color{blue}\xout{$L(\gamma_2)\le K \tau_\epsilon\sim K \epsilon$.}
 \[
 \ell(\gamma_2)\le K \tau_\epsilon\sim K \epsilon.
 \]
Let us construct $\Gamma_2$.
Let $A_{\gamma(0)}$ and $B_{\gamma(0)}$ be closed cones at $T_{\gamma(0)}M$ whose intersection with the hyperplane $\dd x^0=1$ give the same ellipsoid centered in $x^i=0, i=1,\cdots n$, up to a rescaling, and such that  $\dot \gamma(0), \p_0\in \textrm{Int}A_{\gamma(0)} $,  $A_{\gamma(0)} \subset \textrm{Int} B_{\gamma(0)}$ and $B_{\gamma(0)} \subset (\textrm{Int}C)_{\gamma(0)}$. Extend the cones $A$ and $B$ over the coordinate neighborhood by using the affine structure induced by the coordinates. By continuity the mentioned inclusions still hold in a neighborhood of $\gamma(0)$.  In what follows, with slight abuse of notation, we denote with $A$ the cone or its exponential map, and similarly for $B$.
Let us make a linear change of the spatial coordinates so that $A$ and $B$ really become cones with spherical sections in the new coordinates.
Let $\alpha$ and $\beta$ be the slopes of $A$ and $B$ with respect to the hyperplane $\dd x^0=0$. We have $\beta<\alpha$. Since $\gamma$ is differentiable at $0$ and $\dot \gamma^0(0)=1$, we have  and $\tau_\epsilon/\epsilon \to 1$, as $\epsilon \to 0$.  Moreover, for sufficiently small $\epsilon$, $\gamma_2$ is included in the cone $A\subset U$ with origin in $\gamma(0)$. On the $x^0$-axis any point with $0\le x^0<\tau_\epsilon (1-\beta/\alpha)$ reaches with a $B$-causal curve (hence $C$-timelike) every point on the $A$ cone with time coordinate $\tau_\epsilon$, hence $\gamma_2(\tau_\epsilon)$.
If we take $r$ sufficiently small it is the case that
for  sufficiently small $\epsilon$
\[
r\epsilon=
h(0)\epsilon < \tau_\epsilon (1-\beta/\alpha),
\]
thus it is possible to connect $\Gamma_1(0)$ with $\gamma_2(\tau_\epsilon)$ with a coordinate-straight $C$-timelike  curve $\Gamma_2$.
\end{proof}

\section{Geodesics in locally Lipschitz spaces}

The space $(M,F)$ is said to be a {\em locally Lipschitz Lorentz-Finsler space} if $C^\times$ (and hence $C$) is a locally Lipschitz cone structure. Whenever $C^\times$ is a locally Lipschitz cone structure lightlike geodesics can be defined unambiguously on $(M^\times, C^\times)$ as locally achronal continuous causal curves. It has been suggested in \cite{minguzzi17} that {\em  causal geodesics} on $(M, F)$ can be defined as projections of lightlike geodesics of $(M^\times, C^\times)$.

The next result uses \cite[Theorem 2.56]{minguzzi17} and \cite[Proposition 2.27]{minguzzi17} and establishes that in locally Lipschitz proper Lorentz-Finsler spaces there is only one natural notion of causal geodesic.

\begin{theorem} \label{kki}
Let $(M,F)$ be a locally Lipschitz proper Lorentz-Finsler space (such that $F(\p C)=0$). The projections of (locally) achronal continuous causal curves coincide with the (resp.\ locally) maximizing continuous causal curves.
\end{theorem}

\begin{proof}
Every maximizing continuous causal curves is indeed the projection of an achronal  continuous causal curve by \cite[Proposition 2.27]{minguzzi17}. In fact for a locally Lipschitz proper cone structure the various notions of achronality really coincide, cf.\ \cite[page 30]{minguzzi17}. Suppose that a continuous causal curve $\gamma\colon [0,1]\to M$, $\gamma(0)=p$, $\gamma(1)=q$, is the projection of an achronal continuous causal curve $\Gamma\colon [0,1]\to M^\times$. By translational invariance we can assume that $\Gamma(0)=(p,0)$ and $\Gamma(1)=(q,z_1)$ with $z_1\ge 0$, otherwise we can reflect $\Gamma$ on the hyperplane $z=0$. By \cite[Theorem 2.17]{minguzzi17} $\Gamma$ has lightlike tangent wherever it is differentiable, hence almost everywhere. This condition reads $\dot z=\pm F(\dot \gamma)$, hence $0\le z_1\le \ell(\gamma)$.

Assume, by contradiction, that $\gamma$ is not maximizing, $d(p,q)>\ell(\gamma)\ge z_1\ge 0$.
Let $R$ be such that $z_1<R<d(p,q)$.  By  \cite[Theorem 2.56]{minguzzi17} there is a timelike curve $x\colon [0,1]\to M$  with endpoints $p$ and $q$ such that $R<\ell(x)\le d(p,q)$. The curve $X(t)=(x(t),\frac{z_1}{\ell(x)} \ell( x\vert_{[0,t]}) )$ is timelike and connects $\Gamma(0)$ to $\Gamma(1)$ thus $\Gamma$ is not achronal, a contradiction. The locally achronal case is obtained localizing the argument.
%If $z_1=0$ let $R$ be such that $z_1=0<R<d(p,q)$, then by  \cite[{\color{blue} Theorem \xout{Thm.}}\ 2.56]{minguzzi17} there is a timelike curve $x\colon [0,1]\to M$ of endpoints $p$ and $q$ such that $R<\ell(x)\le d(p,q)$. The curve $X(t)=(x(t),0)$ is also timelike and connects $\Gamma(0)$ to $\Gamma(1)$, thus $\Gamma$ is not achronal, a contradiction.
%If $z_1>0$ let $R=z_1<d(p,q)$ then by  \cite[{\color{blue} Theorem \xout{Thm.}}\ 2.56]{minguzzi17} there is a timelike curve $x\colon [0,1]\to M$ of endpoints $p$ and $q$ such that $R<\ell(x)\le d(p,q)$. The curve $X(t)=(x(t),\frac{R}{\ell(x)} \ell( x\vert_{[0,t]}) )$ is timelike and connects $\Gamma(0)$ to $\Gamma(1)$ thus $\Gamma$ is not achronal, a contradiction. The locally achronal case is obtained localizing the argument.
\end{proof}

It is interesting to clarify the connection between the spaces introduced in  Theorem \ref{jol} and the locally Lipschitz Lorentz-Finsler spaces. The spaces of Theorem \ref{jol} are clearly more restrictive because through the function $f$ they constrain the behavior of $F$ near the boundary of the light cone in a way independent of the lightlike vector approached. The spaces of Theorem  \ref{jol}, just slightly strengthened, are indeed locally Lipschitz Lorentz-Finsler  spaces.

\begin{theorem} \label{gud}
The spaces considered in Theorem \ref{jol} for which (b) is replaced by the stronger condition
\begin{itemize}
\item[(b')] the  map $(x,v)\mapsto f(F_x(v))$ is vertically strongly differentiable, in the sense that in a coordinate trivialization there is a linear map $A_{(x,v_1)}\colon \mathbb{R}^{n+1}\to \mathbb{R}$, continuous in $(x, v_1)$, such that\footnote{Due to the continuity of $A$ we  can replace $A_{(x,v_1)}$ for $A_{(\bar x,\bar v)}$ in Eq.\ (\ref{mko}), then we see that this definition coincides with that of partial strong differentiation with respect to $v$ as defined by Nijenhuis \cite{nijenhuis74}. Equation (\ref{mko}) is in a form more closely related to the conditions imposed in Whitney extension theorem \cite{whitney34}.}
\begin{equation}\label{mko}
R(x,v_1,v_2)= \frac{ f(F_x(v_2))-f(F_x(v_1))-A_{(x,v_1)}(v_2-v_1)}{\vert v_2-v_1\vert}
\end{equation}
goes uniformly to zero  as $v_2,v_1\to \bar v$, $x\to \bar x$, and moreover  $A_x=\dd f(F_x)\ne 0$ on $\p C_x$  (thus $\p C_x$ is $C^1$),
\end{itemize}
are locally Lipschitz Lorentz-Finsler spaces.
\end{theorem}

Notice that we have simply added the uniformity of $R$ in $x$. Our feeling is that this condition placed just on $C$ could be sufficient to guarantee the existence of the extension of $f(F)$ used in Theorem  \ref{jol}. Unfortunately, this result would require a certain improvement of Whitney extension theorem  (due to the presence of the additional parameter $x$) that we do not try to elaborate here.

This  proof is an improvement of the proof in \cite[Theorem 2.52]{minguzzi17}. The conditions $(a)$ and $(b')$ are satisfied in the metric locally Lipschitz case thus recovering one direction of \cite[Theorem 2.51]{minguzzi17}.

\begin{proof} Let us prove that $C$ is locally Lipschitz.
Let $\bar x\in M$, and let $U$ be a  coordinate neighborhood of $\bar x$. Let us consider the trivialization of the bundle $T U$, as induced by the coordinates. We are going to focus on the subbundle of $TU$ of vectors that in coordinates read as follows $(x^\alpha, y^\alpha)$ where $y^0=1$, i.e. we are going to work on $U\times \mathbb{R}^n$.
It will be sufficient to prove the locally Lipschitz property for the distribution $x\mapsto S_x$, where $S_x$ is the boundary of the sliced cone at $x$. Let $\vert \cdot \vert$ be the Euclidean norm on $\mathbb{R}^n$. Let  us consider the  function $u:=f\circ F\vert_{\{y^0=1\}}$.
Since the cone distribution over the sliced subbundle has compact fibers, we can find $U$ sufficiently  small such that there is a constant $A>0$,  with $ \vert \nabla_y u\vert>A$ for all lightlike vectors on the sliced subbundle (the labels $x$ and $y$ refer to base and  vertical variables, respectively).

%Let us consider two sliced cones $S_{x_1}$ and $S_{x_2}$. We bound the Hausdorff distance  of the sliced cones by bounding the distance between these boundaries. For any given $y\in S_{x_1}$, the function $d(x_1,x_2)=\sup_{y_1\in S_1}\inf_{y_2\in S_2} \vert y_2-y_1 \vert$ is realized by values $y_1$ and $y_2$ such that the segment $\delta y=y_1-y_2$ is orthogonal to $S_{x_2}$ at $y_2$. The function $d$ is the minimum constant such that $S_{x_1}\subset S_{x_2}+d \bar B^n(0,1)$.

 Let $y_1\in S_{x_1}$ and $y_2\in S_{x_2}$ be two points that realize the Hausdorff distance $D(x_1,x_2)$ between the sliced cone boundaries, i.e. $D(x_1,x_2)= \vert \delta y \vert$, $\delta y=y_1-y_2$.
 %, where the vector $\delta y= y_1-y_2$ can be identified with a vector of $\mathbb{R}^n$ since its 0-th component vanishes.
 The definition of Hausdorff distance easily implies that $\delta y$ is orthogonal to one
of the sliced cone boundaries. Let it be that of $x_2$, the other case being similar. So we have $\delta y\propto \nabla_y u(x_2,y_2)$, and hence $\vert\nabla_y u \cdot \delta y\vert=\vert\nabla_y u \vert \vert \delta y\vert$. Let $\delta x=x_1-x_2$, $u(x_1,y_1)-u(x_2,y_2)=0$. By the continuity of the cone distribution we have $\delta y\to 0$  as $\delta x \to 0$. Moreover,
\begin{align*}
0&=u(x_1,y_1)-u(x_2,y_2)=u(x_1,y_1)-u(x_2,y_1)+u(x_2,y_1)-u(x_2,y_2)
\\
&=u(x_1,y_1)-u(x_2,y_1)+\nabla_y u(x_2,y_2) \cdot \delta y +R(x_2,y_2,y_1) \vert\delta y\vert
\end{align*}
%\begin{align*}
%0&=L(x_1,y_1)-L(x_2,y_2)=L(x_1,y_1)-L(x_1,y_2)+L(x_1,y_2)-L(x_2,y_2)
%\\
%&=\nabla_y L(x_1,y_2) \cdot \delta y +o_{x_1}(\vert\delta y\vert)+L(x_1,y_2)-L(x_2,y_2)
%\end{align*}
%Here $\vert o_{(x_2,y_2)}(\vert\delta y\vert) \vert\le r \vert\delta y\vert$, where $r>0$ is a constant independent of $(x,y)\in S$ (this follows from the uniform convergence of $R$ to zero)
thus
\begin{align*}
\vert \nabla_y u(x_2,y_2) \vert \, \vert \delta y \vert&=\vert \nabla_y u(x_2,y_2) \cdot \delta y \vert \le \vert u(x_2,y_1)-u(x_1,y_1)\vert+ \vert R(x_2,y_2,y_1)  \vert \vert\delta y\vert\\
&\le K \vert \delta x\vert+ \vert R(x_2,y_2,y_1)  \vert \vert\delta y\vert,
\end{align*}
where $K$ is the local Lipschitz constant for $f\circ F$ and hence $u$ over $\bar U$.
The other case in which $\delta y$ is orthogonal to $S_{x_1}$ would have lead to
\begin{align*}
\vert \nabla_y u(x_1,y_1) \vert \, \vert \delta y \vert&=\vert \nabla_y u(x_1,y_1) \cdot \delta y \vert \le \vert u(x_1,y_2)-u(x_2,y_2)\vert+ \vert R(x_1,y_2,y_1)  \vert \vert\delta y\vert\\
&\le K \vert \delta x\vert+ \vert R(x_1,y_2,y_1)  \vert \vert\delta y\vert.
\end{align*}
For every $\epsilon>0$ and $\bar y\in S_{\bar x}$, $\bar x=x_1$, we can find a neighborhood $V\ni \bar x$ and $\delta >0$ such that
for $x'\in V$ and $y',y''$ such that $\vert y'-\bar y\vert<\delta$, $\vert y''-\bar y\vert<\delta$, we have $\vert R(x',y',y'')\vert<\epsilon$.
But $S_{\bar x}$ admits a finite covering with balls of radius $\delta^i/2$ centered at points $y^i\in S_{\bar x}$, $i=1,\cdots, k$, thus given $\epsilon>0$ and $\bar x$ and defined $V=\cap_i V^i$  and $\delta=\min \delta^i/2$,  we have for every $\bar y\in S_{\bar x}$ and $y'$, $\vert y'-\bar y\vert<\delta$, and $x'\in V$, that $\vert R(x',y',\bar y)\vert<\epsilon$. Thus given $\epsilon >0$, $\epsilon\le A/2$, we can choose $x_2$ so close to $x_1$ that $\vert \delta y\vert<\delta$ (due to the continuity of the cone distribution), and $x_2\in V$, which implies both $\vert R(x_1,y_2,y_1)\vert<\epsilon$ and $\vert R(x_2,y_2,y_1)\vert<\epsilon$.

Thus for $x_2$ sufficiently close to $x_1$ we have $D(x_1,x_2) \le \frac{K}{A-\epsilon}\le \frac{2K}{A} \vert \delta x\vert$ where the  constant $\frac{2K}{A}$ does not depend on $x_1 \in U$. In fact we can get the inequality for any $x_2\in U$ connected to $x_1$ by a segment $x(s)=(1-s)x_1+sx_2$ in $U$. We just need  to consider the  maximal interval $I\ni 0$, such that
\[
I\backslash\{0\}\subset \left\{s: \frac{D(x_1,x(s))}{\vert x(s)-x_1 \vert}\le \frac{2K}{A}\right\}.
\]
 We have just proved that $I$ is a non-trivial interval of 0. It is closed due to the continuity of the cone distribution, but it is also open because if $\bar s$ belongs to it then in a small neighborhood of  $\bar s$, for any $s>\bar s$, by the previous argument (independent of $x_1$) $D(x(\bar s),x(s))\le \frac{2K}{A} \vert x(s)-x(\bar s)\vert$, hence
\begin{align*}
D(x_1,x(s))&\le D(x_1,x(\bar s))+D(x(\bar s),x(s))\\
&\le \frac{2K}{A} \{\vert x(\bar s)-x_1\vert+\vert x(s)-x(\bar s)\vert\}=\frac{2K}{A}\vert x( s)-x_1\vert,
\end{align*}
which proves that $I$ is also open. We conclude that $s=1\in I$, and hence, due to the arbitrariness of $x_1$, that for any $x_1,x_2 \in U$, $D(x_1,x_2)\le \frac{2K}{A} \vert x_2-x_1\vert$, i.e.\ $C$ is locally Lipschitz.

%Finally, using $\vert \delta y\vert=D(x_1,x_2)$
%\[
%\limsup_{x_2\to x_1} \frac{1}{\vert \delta x\vert}D(x_1,x_2)\le \frac{K}{A-\epsilon}\le \frac{2K}{A}
%\]
 %for sufficiently small $\delta x$
%\[
%D(x_1,x_2)=\vert \delta y\vert\le (1+K)\vert \delta x\vert
%\]

 %This inequality proves  the locally Lipschitz property (by Stepanov Theorem the function $D(x_1,\cdot)$ is differentiable almost everywhere with differential  uniformly bounded by $2K/A$, hence $2K/A$-Lipschitz).
 It can be observed that this argument used only properties $(a)$ and $(b')$ of $f\circ F$. Now, for the local Lipschitzness of $C^\times$ we need only to show that there is a  function over $ C^\times\cap \{z\ge 0\}$, where $z$ is the extra tangent space coordinate, with the same properties and that vanishes on $C^\times$ (it need not be constructed from a Finsler function $F^\times$ as $u$ is)
 Evidently, the function $-f(F_x(y))+f(z)$, has the desired properties.
\end{proof}

\section{Conclusions}

It is natural to expect that timelike geodesic completeness should prevent the possibility of spacetime extensions. Roughly, the idea is that
any extension would introduce some timelike geodesic crossing the boundary of the original spacetime, which then could not be timelike complete.
By following this type of strategy Galloway, Ling and Sbierski indeed proved that spacetime  inextendibility follows from timelike completeness  \cite{galloway18b}. Unfortunately, they made an unwanted assumption of global hyperbolicity which, with this work, we have shown to be unnecessary.

% they imposed a  quite strong causality condition. In this work we have shown that this unwanted assumption can indeed be dropped.
% and that the proof can be given for the class of $C^0$ proper Lorentz-Finsler spaces.

Our result (Theorem \ref{jif}) states that $C^0$ proper Lorentz-Finsler spaces cannot be extended while preserving this regularity, provided they are timelike complete in a suitable sense.
%has the nice feature that the original spacetime can itself be rough and that the regularity of the possible extension belongs to the same  regularity class.
%Further,
Therefore, the regularity class is the same for the original and the extended spacetimes and, more importantly,  our theorem  holds true for  Finslerian anisotropic spacetimes. Notice that these spacetimes are much more general than Lorentzian spacetimes because of the degrees of freedom entering the shape of the cone distribution and the Finsler function $F$. We have also shown that many satellite results developed in \cite{galloway17b,galloway18b} keep their validity in the Finslerian framework.

While these $C^0$ extendibility studies were motivated by the Cosmic Censorship Conjecture, other motivations have promoted the recent development of non-regular spacetime geometry. One  idea is that the study of low regularity spacetimes  can help to identify those causality concepts that  might preserve their physical significance at a more fundamental (quantum) scale.

It was then of particular importance to generalize Graf and Ling's theorem on the causal character of maximizers to the Lorentz-Finsler case. Interestingly, and compatibly with the above idea, the generalization   has turned out to hold provided a certain  velocity-momentum map $v \mapsto p=\dd L(v)$  is well defined for massless as well as massive particles (condition (b) of Theorem \ref{jol}). The physical significance of this condition confirms that the study of non-regular and Finslerian anisotropic spacetimes provides new insights into the geometry of gravitation.

Finally, we proved that for the locally Lipschitz Lorentz-Finsler spaces in the sense of \cite{minguzzi17} all the definitions of geodesics really coincide (Theorem \ref{kki}), and that the locally Lipschitz spaces of Theorem \ref{jol} for which the causal characterization of maximizers holds belong to that category provided their regularity conditions are only slightly strengthened (Theorem \ref{gud}).

%\bibliography{../../bibliografie/simultaneity,../../bibliografie/libri,../../bibliografie/miei,../../bibliografie/mieiPrep,../../bibliografie/mieiProc}

\begin{thebibliography}{10}

\bibitem{aubin84}
Aubin, J.-P. and Cellina, A.: \emph{Differential inclusions}, vol. 264 of
  \emph{Grundlehren der Mathematischen Wissenschaften [Fundamental Principles
  of Mathematical Sciences]}.
\newblock Springer-Verlag, Berlin (1984)

\bibitem{bernard16}
Bernard, P. and Suhr, S.: Lyapounov functions of closed cone fields: from
  {C}onley theory to time functions.
\newblock Commun. Math. Phys. \textbf{359}, 467--498 (2018).

\bibitem{chrusciel12}
Chru{\'s}ciel, P.~T. and Grant, J. D.~E.: On {L}orentzian causality with
  continuous metrics.
\newblock Class. Quantum Grav. \textbf{29}, {145001} (2012)

\bibitem{chrusciel18}
Chru\'sciel, P.~T. and Klinger, P.: The annoying null boundaries.
\newblock J. Phys.: Conf. Ser. \textbf{968}, 012003 (2018)

\bibitem{galloway17b}
Galloway, G.~J. and Ling, E.: Some remarks on the {$C^0$-}(in)extendibility of
  spacetimes.
\newblock Ann. Henri Poincar\'e \textbf{18}, 3427--3447 (2017)

\bibitem{galloway18b}
Galloway, G.~J., Ling, E., and Sbierski, J.: Timelike completeness as an
  obstruction to {$C^0$-}extensions.
\newblock Commun. Math. Phys. \textbf{359}, 937--949 (2018)

\bibitem{graf18}
Graf, M. and Ling, E.: Maximizers in {L}ipschitz spacetimes are either timelike
  or null.
\newblock Class. Quantum Grav. \textbf{35}, 087001 (2018)

\bibitem{grant18}
Grant, J. D.~E., Kunzinger, M., and S{\"a}mann, C.: Inextendibility of
  spacetimes and {L}orentzian length spaces.
\newblock Ann. Glob. Anal. Geom. \textbf{55}, 133--147  (2019)


\bibitem{kunzinger18}
Kunzinger, M. and S{\"a}mann, C.: Lorentzian length spaces.
\newblock Ann. Global Anal. Geom. \textbf{54}, 399--447 (2018)

\bibitem{minguzzi17}
Minguzzi, E.: Causality theory for closed cone structures with applications.
\newblock Rev. Math. Phys. \textbf{31}, 1930001 (2019).

\bibitem{nijenhuis74}
Nijenhuis, A.: Strong derivatives and inverse mappings.
\newblock Amer. Math. Monthly \textbf{81}, 969--980 (1974)

\bibitem{sbierski15}
Sbierski, J.: The {$C^0$}-inextendibility of the {S}chwarzschild spacetime and
  the spacelike diameter in {L}orentzian geometry.
\newblock J. Diff. Geom. \textbf{108}, 319--378 (2018)

\bibitem{whitney34}
Whitney, H.: Analytic extensions of differentiable functions defined in closed
  sets.
\newblock Trans. Amer. Math. Soc. \textbf{36}, 63--89 (1934)

\end{thebibliography}
%\bibliographystyle{cmp}

\end{document}